%% file: Elia_Jalden_Fundamental_Limits_Arxiv.tex
\documentclass[10pt,twocolumn,twoside]{IEEEtran}

\usepackage{calc}

\usepackage{cite}
\usepackage{graphicx,subfigure}
\usepackage{psfrag}
\usepackage{amsmath,amssymb}
\input defs.tex

\input macros.tex

\def\ml{\mathrm{ML}}
\def\lat{\mathrm{L}}

\def\X{\mathcal{X}}
\def\D{\mathcal{D}}
\def\P{\mathcal{P}}

\def\XD{\mathcal{X},\!\mathcal{D}}
\def\XDP{\mathcal{X},\!\mathcal{D},\!\mathcal{P}}

\newcommand{\beq}{\begin{equation}}
\newcommand{\eeq}{\end{equation}}

\newcommand{\ba}{\begin{array}}
\newcommand{\ea}{\end{array}}
\newcommand{\bea}{\begin{eqnarray}}
\newcommand{\eea}{\end{eqnarray}}
\newcommand{\beqn}{\begin{equation*}}
\newcommand{\eeqn}{\end{equation*}}
\newcommand{\bean}{\begin{eqnarray*}}
\newcommand{\eean}{\end{eqnarray*}}

\newcommand{\limrho}{\lim\limits_{\rho\rightarrow \infty}}

\begin{document}
\sloppy
\title{Fundamental Rate-Reliability-Complexity Limits in Outage Limited MIMO Communications}
\author{
\thanks{The research leading to these results has received funding from the European Research Council under the European Community's Seventh Framework Programme (FP7/2007-2013) / ERC grant agreement no. 228044. P. Elia acknowledges funding by EU:FP7/2007-2013 grant no. 216076 (SENDORA), by NEWCOM++ (contract n. 216715), and by Mitsubishi RD project Home-eNodeBS. The work of J. Jald{\'e}n was supported by the SSF grant ICA08-0046.}
Petros Elia\thanks{P. Elia is with the Mobile Communications Department, EURECOM, Sophia Antipolis, France (email: elia@eurecom.fr)} and Joakim Jald{\'e}n\thanks{J. Jald{\'e}n is with the ACCESS Linnaeus Center, KTH Signal Processing Lab, Royal Institute of Technology, Stockholm, Sweden (email:jalden@kth.se)} }

\maketitle
\begin{abstract}
The work establishes fundamental limits with respect to rate, reliability and computational complexity, for a general setting of outage-limited MIMO communications.  In the high-SNR regime, the limits are optimized over all encoders, all decoders, and all complexity regulating policies.
The work then proceeds to explicitly identify encoder-decoder designs and policies, that meet this optimal tradeoff.
In practice, the limits aim to meaningfully quantify different pertinent measures, such as the optimal rate-reliability capabilities per unit complexity and power, the optimal diversity gains per complexity costs, or the optimal number of numerical operations (i.e., flops) per bit.
Finally the tradeoff's simple nature, renders it useful for insightful comparison of the rate-reliability-complexity capabilities for different encoders-decoders.
\end{abstract}

\section{Introduction}
\subsection{General system model}
We consider the general multiple-input multiple-output (MIMO) communications setting,
where the $m\times 1$ vector representation of the received signal $\bfy$ is given by
\beq \label{eq:generalMIMO}
\bfy=\bfH\bfx+\bfw,
\eeq
where $\bfx$ is the $n\times 1$ vector representation of the coded transmitted signals, $\bfH$ the $m\times n$ channel matrix, and where $\bfw$ represents additive noise.  $\bfH$ is considered to be random, having an arbitrary distribution, and being parameterized by $\rho$ which is interpreted as the SNR (cf. \cite{JE:09b}).  $\bfw$ is taken to be i.i.d.\ Gaussian with fixed variance.  We assume that one use of \eqref{eq:generalMIMO} corresponds to $T$ uses of some underlying ``physical'' channel.

The model applies to several network topologies and scenarios, such as \textsc{mimo, mimo-ofdm, mimo-mac, mimo-arq}, and cooperative communications, where each such scenario endows $\bfH$ and $\bfx$ with different structures, dimensionalities and statistics.  This work specifically considers the non-ergodic, outage-limited setting, in which the above MIMO-related scenarios play a crucial role in improving the error and rate performance, though usually at the expense of much higher encoding-decoding computational complexity.

\subsection{Motivation and general results}
Error performance and encoding-decoding complexity in telecommunications (cf. \cite{ZT:03,Tel:99,AEV:02,BK:98,Mic:01}), are widely considered to be two limiting, and interrelated bottlenecks.  Joint exposition of these two aspects becomes increasingly necessary, in order to meaningfully quantify the ever increasing complexity costs of reliable communication, in systems that progressively become larger and more dynamic.

A natural question then relates to establishing and meeting joint fundamental error-performance and complexity limits, optimized over all choices of encoders, decoders and policies.
Such limits will be here described, under a high SNR approximation, in the form of an optimal rate-reliability-complexity tradeoff for MIMO communications.

The limits provide answers on the pertinent measure
\[\bigl(\text{\textsc{snr}, rate, reliability, complexity}\bigr),\]
and do so within approximation factors which, for increasing $\rho$, vanish to be smaller than any $\rho^\epsilon$, for any $\epsilon>0$.  Specifically these answers pertain to the following.
\begin{itemize}
\item Description of the best achievable \textsc{snr}-rate-reliability-complexity combination, optimized over all transceivers and policies (Theorem~\ref{thm:optimalPair}).
\item Description of the union of all achievable \textsc{snr}-rate-reliability-complexity combinations. (Corollary~\ref{thm:achievableRegion}).
\item Description of the optimal value achieved by a large family of utility measures which quantify the \textsc{snr}-rate-reliability-complexity capabilities of transceivers, and which are decreasing functions of complexity and of error probability.  (Corollary~\ref{thm:OptimalUtility}).
\end{itemize}

\subsection{Structure of paper}
Section~\ref{sec:transceiverDesign} recalls the general transceiver setting, and defines the different performance measures.
Section~\ref{sec:exponents} introduces the asymptotic measures of performance, directly applying the \emph{diversity multiplexing tradeoff} (DMT, \cite{ZT:03}) as the pertinent asymptotic measure of rate-error performance, and defining the \emph{worst-case complexity exponent} as a high-SNR asymptotic measure of the worst-case complexity of (reasonable) transceivers.  Section~\ref{sec:tradeoff} presents the optimal high-SNR rate-reliability-complexity tradeoff, and the optimal transceiver utility value in its general form, as well as in its simpler, more specific,  \emph{homogeneous} variant.  Finally Section~\ref{sec:conclussions} concludes.

\section{Transceiver design and decoding policy: \\\textsc{snr}, rate, reliability and complexity\label{sec:transceiverDesign}}

\subsection{Transceiver design and decoding policy}

Consider a sequence of transceiver designs $\X_\rho,\D_{\rho}$, parameterized by $\rho$, where $\X_\rho\subset \reals^{n}$ denotes the codebook that maps information into transmitted signals, and where $\D_\rho$ denotes the decoder(s) that extract information from the received signals.
Let the transmitted codewords $\bfx$ be picked, with uniform probability, from the codebook $\X_{\rho} $.  Transmission has duration $T$, SNR $\rho$, rate \[R = \frac{1}{T} \log_{2} |\Xset_\rho|,\] and an enforced power constraint such that \beq\label{eq:powerConstraint}
\frac{1}{|\X_\rho|}\sum\limits_{\bfx\in \X_\rho}\|\bfx\|^2=T.
\eeq
For simplicity we write $\XD$, and we let the parameterization be implied.

Consider a policy $\P$ (short for $\P_{\rho,\!\XD}$), which generally trades-off error performance with complexity, by forcing the decoder to limit the number of numerical operations (i.e., flops), up to a maximum designated number of flops.  Once this limiting number of flops is reached, the decoder quits and declares an error. This limiting number of flops may or may not be chosen as a function of the instantaneous $\bfH, \bfy$, and will generally depend on $\rho$.

\subsection{Rate, reliability and complexity}
The error probability $P_{{\XDP}}$ introduced by the specific $\XDP$, is simply
\beq \label{eq:errorProbability}
P_{{\XDP}} := \prob{\bfH,\bfx,\bfw \ : \  \bfxh_{\XDP} \neq \bfx},
\eeq
where $\bfxh_{\XDP}$ denotes the vector decoded by $\D$, under the restrictions of $\P$.
For a given $\XDP$ and a given realization of problem inputs $\bfH,\bfx,\bfw$, then $N_{\XDP}(\bfH,\bfx,\bfw)$ will denote the overall instantaneous introduced complexity, in flops.  Then worst-case complexity is simply given by
\beq \label{eq:errorProbability}
C_{{\XDP}} := \sup_{\bfH,\bfx,\bfw}N_{\XDP}(\bfH,\bfx,\bfw).
\eeq
A pertinent measure of performance for any $\XDP$ then becomes the corresponding set of achievable combinations $(\rho, R, P_{{\XDP}},C_{{\XDP}})$, or its equivalent one-to-one re-mapping \[\bigl(\rho, R, \frac{\log P_{{\XDP}}}{\log Z},\frac{\log C_{{\XDP}}}{\log L}\bigr),\] where $Z,L$ regulate the refinements of the representation.

\section{Error and complexity exponents\label{sec:exponents}}

\subsection{Quantifying error performance: DMT}
As a measure of rate-reliability performance, we adopt the refinement of the diversity-multiplexing tradeoff, identified by Zheng and Tse in \cite{ZT:03}, as a fundamental performance limit in outage-limited MIMO communications.

In this setting, both the error probability $P_{{\XDP}}$ introduced by the specific $\XDP$, as well as the cardinality of $\X$, are parameterized by $Z=\rho$.  Specifically the code cardinality \[|\X|=2^{RT},\] is described by the \emph{multiplexing gain}
\beq \label{eq:multiplexing-gain}
r \defeq \lim_{\rho \rightarrow \infty} \frac{R}{\log_{2} \rho} = \lim_{\rho \rightarrow \infty} \frac{1}{T} \frac{\log |\X|}{\log \rho} \, ,
\eeq
and the associated error performance delivered by the transceiver and policy, is described by
the limiting behavior of $\frac{\log P_{{\XDP}}}{\log \rho}$, i.e., by the \emph{diversity gain}~\cite{ZT:03}
\bea
\label{eq:dmtXDP}d_{\XDP}(r)\!\!\!\!&:=&\!\!\!\!-\limrho \frac{\log P_{{\XDP}}}{\log\rho}.
\eea

\subsection{Regulating and quantifying complexity performance: worst-case complexity exponent}
We now consider the one-to-one mapping 
\[C_{{\XDP}} \ \leftrightarrow \ \frac{\log C_{{\XDP}}}{\log L},\] where $L$ is a properly chosen scaling factor, being for example a function of $|\X|$.  Seeking for patterns and insight, we move to asymptotics where a general asymptotic worst-case complexity measure then takes the form
\beqn
\limrho\frac{\log C_{{\XDP}}}{\log L}.
\eeqn
Similar to the DMT in \cite{ZT:03} which measures the high-SNR $P_{\XDP}$ as a polynomial power of $\rho$, the currently chosen measure of complexity will also be an exponent over $L=\rho$, 
taking the form
\beq \label{eq:complexityExponent}
c_{\XDP}(r) : =   \limrho \frac{\log C_{\XDP}}{\log\rho}.
\eeq
We note that the chosen worst-case complexity exponent $c_{\XDP}(r)$ keeps in line with the relevant behavior of most known transceivers, uniformly covering the full complexity range 
\[0\leq c_{\XDP}(r)\leq rT \]
of all reasonable transceivers, with $c_{\XDP}(r)=0$ corresponding to the fastest possible transceiver (requiring a small fixed number of flops per codeword), and with $c_{\XDP}(r)= rT$ corresponding to the slowest, full-search uninterrupted ML decoders\footnote{We here note that strictly speaking, $\XDP$ may potentially introduce a complexity exponent larger than $rT$.  In such a case though, $\XDP$ may be substituted by a lookup table implementation of $\X$ and an unrestricted ML decoder.  This encoder-decoder will jointly introduce a worst case complexity that is a constant multiple of $|\X|\doteq \rho^{rT}$.  It is noted that the number of flops per visited codeword is independent of $\rho$.} in the presence of a canonical code with multiplexing gain $r$, i.e., with $|\X| \doteq \rho^{rT}$.  In the above, the $\doteq$ notation is used when $f(\rho) \doteq \rho^{x}$ iff (cf. \cite{ZT:03})
\begin{equation} \label{eq:doteq}
\lim_{\rho \rightarrow \infty} \frac{\log f(\rho)}{\log \rho} = x \, ,
\end{equation}
and the symbols $\dotgeq$ and $\dotleq$ are defined similarly.

We also note that both $d_{\XDP}(r)$ and $c_{\XDP}(r)$ quantify worst-case (non-ergodic) behavior, and they are both set by the structural properties of the design $\XDP$ as well as the statistical properties of $\bfH,\bfx,\bfw$.

We proceed to find the optimal $\text{opt}_{\XDP} \bigl(d_{\XDP}(r),c_{\XDP}(r)\bigr)$ which is equivalent to finding the optimal $\text{opt}_{\XDP} \bigl(\rho, R, P_{{\XDP}},C_{{\XDP}}\bigr)$ up to a factor that vanishes to a number smaller than any $\rho^\epsilon$, for any $\epsilon>0$, in the limit of high-SNR.

\section{Performance-complexity tradeoff\label{sec:tradeoff}}

We proceed to establish the fundamental limits, optimized over all achievable \textsc{snr}-rate-reliability-complexity combinations of any transceiver and policy, up to a factor that vanishes in the limit of high $\rho$.

Towards this we describe the decoder and encoder structures, that together with a specific policy, meet a natural upper bound to this tradeoff, for all values of $r$.  We start with the decoder, but for now disregard the policy.

\subsubsection{The candidate decoder -- the DMT optimal LLL based LR-aided, regularized linear decoder}
We focus on the efficient and DMT optimal, LLL-based lattice-reduction (LR)-aided regularized linear decoder, presented in its general form in \cite{JE:09a,EJ:10a,JE:09b} for different settings, drawing from works such as \cite{YW:02},\cite{WF:03}.  We clarify that the decoder applies to lattice codes, and for completeness recall the decoder's three main steps.  In the first step, the decoder performs \emph{regularization }via MMSE-GDFE like preprocessing, thus inducing a regularized metric (cf.\cite{JE:09b})
\beq\label{eq:lattice-decoder}
\bfxh_{\lat} = \arg \min_{\bfxh \in \Lambda_{r}}  \|Ê\bfy - \bfH \bfxh \|^2 + \|Ê\bfxh \|^2_{\bfT}.
\eeq
In the above, $\Lambda_{r}$ is the scaled lattice corresponding to the code, and $\bfT$ is a positive definite matrix.  The above metric penalizes far away elements of $\Lambda_{r}$ that are generally non codewords.
The second step includes \emph{lattice-reduction} using the LLL algorithm \cite{LLL:82}, and the last step is an efficient \emph{linear detection }using, for example, the \emph{rounding off} algorithm.

Under standard assumptions on continuity, and in the presence of a policy $\P_{rT}$ that lets the decoder run its course irrespective of the complexity, the above decoder was shown in \cite{JE:09b} to be DMT optimal, i.e., that
\[ d_{\X,\D_{\text{LRR}},\P_{rT}}(r) = \sup\limits_{\D}d_{\X,\D,\P_{rT}},\]
irrespective of the lattice design $\X$, and irrespective of the fading statistics.

It is the case though that the decoder's LLL step introduces worst-case complexity that is infinite \cite{JSM:08}.  This problem is successfully addressed by the policy discussed below.
\subsubsection{The LR-based policy $\P_{\textsc{lr}}$}To limit the above infinite complexity, the work in \cite{JE:09b} proposed a policy that capitalizes on the fact that to achieve DMT optimality, it is not required to LLL reduce every conceivable channel.  Instead, in the event that too many flops occur, the policy instructs the implementation of the LLL algorithm to halt, and the decoder to declare an error.  Special emphasis is given to guaranteeing that the event of halting is not more common than the event of error, thus avoiding degradation of the asymptotic error performance.
Specifically the halting policy, to be denoted as $\P_\textsc{lr}$, was defined on the basis of the bound on the number $K$ of LLL cycles that are necessary for reduction of matrix $\bfM$ which generates the composite code-channel lattice.  This bound is given by \cite{JSM:08,DV:94} to be
\begin{equation} \label{eq:LLL-bound}
K \leq n^2 \log_{\frac{2}{\sqrt{3}}} \kappa(\bfM) + n,
\end{equation}
where $\kappa(\bfM)$ denotes the 2-norm condition number of $\bfM$.  Based on this bound, $\P_\textsc{lr}$ deploys the LLL algorithm only if
\beq \label{eq:LRboundFlops}\kappa(\bfM) \leq \rho^{\frac{1}{2}(d_{\ml}(r) + 1)+\epsilon}, \  \ \epsilon>0,\eeq  where $d_{\ml}(r)\defeq d_{\X,\D_{\text{ML}},\P_{rT}}(r)$ describes the DMT achieved by the uninterrupted ML decoder.  By showing that $$\prob{\kappa(\bfM) \geq \rho^{\frac{1}{2}(d_{\ml}(r) + 1)+\epsilon} } \dotleq \rho^{-d_{\ml}(r)},$$ i.e., that the event of halting is less common than the event of error under full ML decoding, it was proven in \cite{JE:09b} that, over any range of multiplexing gains $r$, the combination of $\D_{\textsc{lrr}}$ and $\P_{\textsc{lr}}$ achieves DMT optimal decoding of any lattice design $\X_\Lambda$, and does so with worst-case complexity of $O(\log \rho)$.  This implies a worst-case complexity that is at most linear in the rate\footnote{The result is extended in \cite{EJ:10a} to the MIMO-MAC case, to show that this optimality holds with worst-case complexity that is at most linear in the users' sum-rate.}, at high SNR.  It also constitutes substantial improvement over sphere decoding implementations where the worst-case complexity reported (see for example \cite{BHV:09} for fast decodable codes \cite{TK:02,PGG:07,SF:07}) is also exponential in $R$, albeit with a smaller exponent than full search.

\subsubsection{The overall worst-case complexity exponent jointly introduced by lattice encoding, $\D_{\textsc{lrr}}$ and $\P_{\textsc{lr}}$}

With the above in mind, we proceed to establish the overall computational complexity jointly introduced by lattice encoding and by the different components of $\D_{\textsc{lrr}}$, in the presence of $\P_{\textsc{lr}}$.

\paragraph{Decoder and policy}
We first quickly note that the regularization and linear-decoding steps, introduce complexity that is essentially independent of $\rho,\bfH$, and bounded above by $O(n^2)$, thus inducing a zero complexity exponent.

Regarding the lattice reduction step, we recall the hard bound
\bean
K & \leq & n^2 \log_{\frac{2}{\sqrt{3}}} \kappa(\bfM) + n \\ & \dot \leq & n^2 \log_{\frac{2}{\sqrt{3}}} \rho^{\frac{1}{2}(d_{\ml}(r) + 1)+\epsilon}+n, \ \ \epsilon>0,\eean
on the number of LLL flops enforced by $\P_{\textsc{lr}}$.  This bound implies that
\[\exists z\in \mathbb{R}^+:\prob{N(\bfH,\bfx,\bfw)>z\log\rho}= 0,\] which in turn means that
\beq \label{eq:LLLbound1}\prob{N(\bfH,\bfx,\bfw)>\rho^{c}}\doteq \rho^{-\infty},  \ \forall c>0.\eeq

In conjunction with the equivalent representation (drawing from \cite{JE:10b}, which presents some $c(r)$ of different $\XDP$)
\beqn \label{eq:complexityExponent2}
c(r)  =  \sup \{c \! : \! -\limrho \frac{\log P( N(\bfH,\bfx,\bfw)\geq \rho^{c} )}{\log\rho} \leq d(r)\}\eeqn
of a worst-case complexity exponent $c(r)$ that allows for $d(r)$, we conclude that the LLL algorithm under $\P_\textsc{lr}$, also introduces an effective complexity exponent equal to zero.  Consequently the entire $\D_{\textsc{lrr}},\P_{\textsc{lr}}$ introduces a minimal complexity exponent, equal to zero.
\paragraph{Lattice encoding}
Moving on to encoding, it is again easy to see that any lattice code $\X_\Lambda$ comes with encoding complexity that is bounded as $O(n^2)$, thus minimally adding to the overall complexity exponent of any transceiver/policy.

We are now able to combine the complexities from the encoder and the decoder, and to provide the following.
\begin{lemma}\label{lem:LRRoptimalityComp}
A lattice code $\X_\Lambda$, in conjunction with the decoder-policy $\D_\textsc{lrr},\P_\textsc{lr}$, jointly accept a minimum, over all encoders, decoders and policies, effective complexity exponent, i.e.,
\beq \label{eq:LRRptimalityComp}c_{\X_\Lambda,\D_\textsc{lrr},\P_\textsc{lr}}=\inf\limits_{\X,\D,\P}c_{\X,\D,\P} =0.
\eeq
\end{lemma}
\subsubsection{The overall error performance}
With respect to the error performance of $\D_\textsc{lrr},\P_\textsc{lr}$, we utilize the result in \cite{JE:09b} which proves that the DMT optimality of $\D_\textsc{lrr},\P_\textsc{lr}$, holds irrespective of the lattice code that it is applied to, i.e., that for \emph{any} fixed lattice code $\X_\Lambda$, then
\beq \label{eq:LRRptimalityDMT}d_{\X_\Lambda,\D_\textsc{lrr},\P_\textsc{lr}}(r)=\sup\limits_{\D,\P}d_{\X_{\Lambda},\D,\P}(r).
\eeq
Disregarding for now issues on code design, we proceed to formalize the performance-complexity optimality of $\D_{\textsc{lrr}}, \P_{\textsc{lr}}$.
\subsubsection{The overall effective complexity/error exponent jointly induced by lattice encoding, $\D_{\textsc{lrr}}$ and $\P_{\textsc{lr}}$}

Combining \eqref{eq:LRRptimalityComp} and \eqref{eq:LRRptimalityDMT} gives the following.
\begin{lemma}\label{lem:LRRoptimalityDMCT}
The high-SNR rate-reliability-complexity tradeoff achieved by the $\D_\textsc{lrr},\P_\textsc{lr}$, is better or asymptotically equal to the tradeoff achieved by any other decoder-policy, irrespective of the lattice code $\X_\Lambda$ applied, i.e.,
\begin{multline}\label{eq:LRRptimalityDMCT}
\bigl(d_{D_{\X_{\Lambda},\D_\text{LRR},\P_\text{LR}}}(r),c_{D_{\X_{\Lambda},\D_\text{LRR},\P_\text{LR}}}(r)\bigr) \nonumber \\ = \bigl(\sup\limits_{\D,\P} d_{\X_\Lambda,\D,\P}(r),\inf\limits_{\D,\P} c_{\X_\Lambda,\D,\P}(r)\bigr).\end{multline}
\end{lemma}
Here it is stressed that this achievable tradeoff may be suboptimal, as it is limited by the reliability capabilities of the specific code $\X_\Lambda$.

What remains now is to combine the optimal components $\D_\textsc{lrr},\P_\textsc{lr}$, with suitable code designs.
\subsubsection{Employing DMT optimal codes, to meet the rate-reliability-complexity tradeoff}
We have just seen in Lemma~\ref{lem:LRRoptimalityDMCT} that, given any lattice design $\X_\Lambda$, the combination $\D_\text{LRR},\P_\text{LR}$ achieves the highest allowable tradeoff over any transceiver-policy that includes $\X_\Lambda$.  Consequently what remains is to identify \emph{lattice }code designs that optimize both $c_{\X_\Lambda,\D,\P}(r)$ and $d_{\X_\Lambda,\D,\P}(r)$, in the presence of $\D_\text{LRR},\P_\text{LR}$.  Optimizing of $c_{\X_\Lambda,\D,\P}(r)$ has already been achieved in Lemma~\ref{lem:LRRoptimalityComp} which proved that any lattice design $\X_\Lambda$ gives $c_{\X_\Lambda,\D_\textsc{lrr},\P_\textsc{lr}}=\inf\limits_{\X,\D,\P}c_{\X,\D,\P} =0.$
Hence what remains is to find a lattice design that optimizes $d_{\X_\Lambda,\D,\P}(r)$, in the presence of $\D_\text{LRR},\P_\text{LR}$.  This in turn is further simplified in the presence of $\eqref{eq:LRRptimalityDMT}$, and the task is now limited to simply finding DMT optimal lattice codes, i.e., codes that asymptotically meet the outage region
\beqn
\mathcal{O} = \{\bfH: \frac{1}{T}\log\det\bigl( I + \beta \bfH\bfH^{\dag} \bigr)<R\},  \ \text{some fixed} \ \beta,
\eeqn
of the equivalent MIMO channel to achieve asymptotically optimal performance (cf.~\cite{ZT:03})\beq \label{eq:DMToutage} d_{\text{opt}}(r) := \sup\limits_{\XDP}d_{\XDP}(r) = P(\bfH\in \mathcal{O}).\eeq  The existence of such lattice codes has been proven in~\cite{GCD:04}, for the quasi-static Rayleigh fading channel, and a unified family of such codes was explicitly constructed in \cite{EKP:06} using \emph{cyclic division algebras} (CDA).  Further such codes have, over the last few years, been described for a plethora of MIMO models.  These codes are based on different variants of CDA codes (cf. \cite{SRS:03},\cite{BR:03}), and have been shown, under basic continuity conditions, to provide DMT optimality for all channel dimensions, and most often for all fading statistics.  Such codes can, for example, be found in \cite{EKP:06,YB:07a,Lu:08,KC:09,EK:09,LH:09,EVA:09,PKE:09}, and they DMT-optimally apply to several MIMO scenarios, including \textsc{mimo, mimo-ofdm, mimo-mac} (Rayleigh fading), \textsc{mimo-arq}, as well as to most existing cooperative communication protocols.

For all the above MIMO scenarios, we have now the final result, which holds under basic continuity conditions.

\subsection{The optimal tradeoff}

\begin{theorem} \label{thm:optimalPair}
The high-SNR optimal, over all encoders, decoders and policies, rate-reliability-complexity behavior is given by
\begin{multline}
\label{eq:optimalityDMCT}
\text{opt}_{\XDP} \bigl(d_{\XDP}(r),c_{\XDP}(r)\bigr) \\
= \bigl(d_{\X_\text{CDA},\D_\text{LRR},\P_\text{LR}}(r),c_{\X_\text{CDA},\D_\text{LRR},\P_\text{LR}}(r)\bigr) = \bigl(d_\text{opt}(r),0\bigr)
\end{multline}
and is achieved for all multiplexing gains, all channel dimensions and (in most known cases) all fading statistics, by the CDA-based designs $\X_\textsc{cda}$, the LR-aided regularized linear decoder $\D_\textsc{lrr}$, and the LR-based policy $\P_\textsc{lr}$.
\end{theorem}

Equivalently the result shows that the achievable rate-reliability-complexity combination
\beq \label{eq:bestPair}\bigl(\rho, R=r\log\rho, P\doteq\rho^{-d_\text{opt}(r)},C\doteq\rho^0\bigr)\eeq is optimal, up to a factor that asymptotically becomes smaller than any $\rho^\epsilon$, for any $\epsilon>0.$.  We quickly note that $\X_\text{CDA},\D_\text{LRR},\P_\text{LR}$ is currently the only known tradeoff-optimal design.

Directly from the above, we have the following.
\begin{corollary}
\label{thm:achievableRegion}
In the high SNR regime, the union of all achievable \textsc{snr}-rate-reliability-complexity combinations, considering all reasonable $\XDP$, is given by
\begin{multline*}\{\bigl(\rho, R=r\log\rho, P\doteq\rho^{-d(r)},C\doteq\rho^{c(r)}\bigr)\}, \\ \ 0\leq d(r)\leq d_\text{opt}(r), 0\leq c(r)\leq rT.\end{multline*}
\end{corollary}
\begin{proof}
For a given $R$, any of the above reliability-complexity pairs can be achieved by employing an $\XDP$ that is optimal with respect to \eqref{eq:bestPair}, modifying though $\P$ to introduce the appropriate amount of extra complexity and errors\footnote{Constructing such modification is trivial.  We note that the worst case ($r,d(r)=0, c(r)=rT$) corresponds to a full-search transceiver that provides subexponential decay of the probability of error, for increasing SNR.
}.
\end{proof}

Finally, using the fact that the complexity of the optimal transceiver is $O(\log\rho)$, it is easy to show that for several MIMO settings, optimal DMT performance is achieved with at most $O(n^2)$ flops per bit.

\paragraph{Optimal limits on general reliability-complexity functions}
Another measure of the rate-reliability-complexity capabilities of different transceivers can take the form of general utility functions.  Towards this we define the following.
\begin{definition}
Let $\Gamma$ be a weighting function that is increasing in $d_{\XDP}(r)$, decreasing in $c_{\XDP}(r)$, and which reflects the different costs assigned separately to erroneous detection, and complexity.
Then we use \beq \label{eq:dmct1}D_{\XDP}(r):=\Gamma\bigl(d_{\XDP}(r),c_{\XDP}(r)\bigr),\eeq
to denote the \emph{$\Gamma$-general rate-reliability-complexity limit}, for a given $\XDP$.
\end{definition}

Towards motivating meaningful use of the limit, we identify the following simple manifestation as one of many special cases of the general limit.
\begin{definition}
The \emph{homogeneous rate-reliability-complexity limit} for a given $\XDP$, and a given weighting factor $\gamma\geq 0$,
takes the form
\beq \label{eq:dmctDiff}D_{\XDP}(r):=d_{\XDP}(r)-\gamma c_{\XDP}(r),\eeq
and describes the diversity gain minus the normalized complexity cost.
\end{definition}

It is interesting to interpret the rate-reliability-complexity limit
$D_{\XDP}(r)$, as a limit that describes the high-SNR error capabilities of $\XDP$, per unit power and complexity. Equivalently, the limit may be described as a measure of diversity gain per complexity order.

The following result, which holds under basic continuity conditions, for the same scenarios as Theorem~\ref{thm:optimalPair}, describes the optimizing value achieved by a large family of measures~$\Gamma$.

\begin{corollary}\label{thm:OptimalUtility}
The optimal, over all encoders, decoders and policies, $\Gamma$-general rate-reliability-complexity limit $D(r)$, is given by
\bea
\label{eq:optimalityDMCT}
D(r) = \Gamma\bigl(\sup\limits_{\XDP} d_{\XDP}(r),0\bigr) = \Gamma\bigl(d_\text{opt}(r),0\bigr)
\eea
and is achieved for all multiplexing gains, and all channel dimensions by the CDA-based designs $\X_\textsc{cda}$, the LR-aided regularized linear decoder $\D_\textsc{lrr}$, and the LR-based policy $\P_\textsc{lr}$.
\end{corollary}
\begin{proof}
The proof is direct by noting that
\bea \label{eq:measure2}
D(r)&=&\sup\limits_{\XDP}\Gamma\bigl(d_{\XDP}(r),c_{\XDP}(r)\bigr) \\ &\leq& \Gamma\bigl(\sup\limits_{\XDP}d_{\XDP}(r),\inf\limits_{\XDP}c_{\XDP}(r)\bigr),\eea and then by applying Theorem~\ref{thm:optimalPair}.
\end{proof}
The following holds for the more intuitive, cost-symmetric version of the limit.
\begin{corollary}
The optimal, over all encoders, decoders and policies, homogeneous rate-reliability-complexity limit, is given by
\beq
\label{eq:optimalityDMCT}
D(r) = \sup\limits_{\XDP} d_{\XDP}(r)-\gamma c_{\XDP}(r)=d_\text{opt}(r).
\eeq
\end{corollary}

\section{Conclusions\label{sec:conclussions}}

The tradeoff and its achievability, provide worst-case guarantees on the complexity required for provably optimal performance in outage-limited MIMO communications.  The guarantees hold over a surprisingly broad setting, and they come with reduced transmission energy and delay, as well as reduced algorithmic power consumption and hardware.
The tradeoff concisely quantifies these guarantees and the capabilities of different transceivers, as well as quantifies the role of policies in simplifying algorithms which would otherwise introduce unbounded complexity.

\bibliographystyle{IEEEtran}
\bibliography{IEEEabrv,refs}

\end{document}

%% file: defs.tex
\def\reals{{\mathbb{R}}}

\def\bfw{{\boldsymbol{w}}}

\def\bfx{{\boldsymbol{x}}}
\def\bfxh{{\hat{\boldsymbol{x}}}}

\def\bfy{{\boldsymbol{y}}}

\def\bfH{{\boldsymbol{H}}}

\def\bfM{{\boldsymbol{M}}}

\def\bfT{{\boldsymbol{T}}}

\def\Xset{\mathcal{X}}

%% file: macros.tex
\newcommand{\prob}[1]{\mathrm{P}\left(#1\right)}

\DeclareMathOperator*{\dotleq}{\overset{.}{\leq}}
\DeclareMathOperator*{\dotgeq}{\overset{.}{\geq}}
\DeclareMathOperator*{\defeq}{\triangleq}

\newtheorem{theorem}{Theorem}
\newtheorem{corollary}{Corollary}[theorem]

\newtheorem{lemma}{Lemma}
\newtheorem{definition}{Definition}